\documentclass[conference,twocolumn,A4paper]{IEEEtran}

\usepackage{virgil_env} 
\usepackage{multirow}

\definecolor{mygray}{RGB}{153,153,153}
\definecolor{myred}{RGB}{221,77,57}
\definecolor{myblue}{RGB}{68,104,252}
\definecolor{mygreen}{RGB}{86,149,22}

\newcommand*{\X}{\ensuremath{X_{1 \ldots n}\xspace}}

\newcommand*{\bin}[1]{\texttt{#1}}

\newcommand*{\Ialpha}{\ensuremath{\opname{I}_{\alpha} \xspace}}

\newcommand*{\Ialphae}[2]{\ensuremath{\opname{I}_{\alpha}\!\left(\{ #1 \} : #2 \right)}}

\usepackage{graphicx}
\usepackage{tikz}
\usepackage{comment}
\usetikzlibrary{calc}

\usepackage{authblk}
\author[1]{Virgil Griffith}
\author[2]{Tracey Ho}
\affil[1]{Computation and Neural Systems, Caltech, Pasadena, CA 91125}
\affil[2]{Computer Science and Electrical Engineering, Caltech, Pasadena, CA 91125}

\usepackage{wasysym}

\usepackage{mathenv}
\usepackage{wasysym}

\title{Quantifying Redundant Information in Predicting a Target Random Variable}

\newtheorem{lem}{Lemma}

\newcommand{\Szero}{$\mathbf{\left(S_0\right)}$\xspace}
\newcommand{\Mzero}{$\mathbf{\left(M_0\right)}$\xspace}
\newcommand{\Mone}{$\mathbf{\left(M_1\right)}$\xspace}

\renewcommand*{\vee}{ {\, , \;} }

\renewcommand*{\vee}{ , }

\begin{document}

\maketitle

\begin{abstract}
This paper considers the problem of defining a measure of redundant information  that quantifies how much common information two or more random variables specify about a target random variable.  We discussed desired properties of such a measure, and propose new measures with some desirable properties.
\end{abstract}

\section{Introduction}

Many molecular and neurological systems involve multiple interacting factors affecting an outcome synergistically and/or redundantly. Attempts to shed light on issues such as population coding in neurons, or genetic contribution to a phenotype (e.g. eye-color), have motivated various proposals to leverage principled information-theoretic measures for quantifying informational synergy and redundancy, e.g.~\cite{berry03,narayanan05,balduzzi-tononi-08,anas2007,lizier13}. In these settings, we are concerned with the statistics of how two (or more) random variables $X_1,X_2$, called predictors, jointly or separately specify/predict another random variable $Y$, called a target random variable. This focus on a target random variable is in contrast to
Shannon’s mutual information which quantifies statistical dependence between two random variables, and various notions of common information, e.g.~\cite{gacs73,Wyner75,KumarLG14}.

The concepts of synergy and redundancy are based on several intuitive notions, e.g., positive informational synergy indicates that $X_1$ and $X_2$ act cooperatively or  antagonistically to influence $Y$; positive redundancy indicates there is an aspect of $Y$ that $X_1$ and $X_2$ can each separately predict. However, it has proven challenging\cite{qsmi,polani12,bertschinger12,Iwedge} to come up with precise information-theoretic definitions of synergy and redundancy that are consistent with all intuitively desired properties.

\section{Background: Partial Information Decomposition}
The \emph{Partial Information Decomposition} (PID) approach of \cite{plw-10} defines the concepts of synergistic, redundant and unique information in terms of  \emph{intersection information}, $\Icape{X_1, \ldots, X_n}{Y}$, which quantifies the common information that each of the $n$ predictors
$X_1, \ldots, X_n$ conveys about a target random variable $Y$. An antichain lattice of redundant, unique, and synergistic partial informations is built from the intersection information.

Partial information diagrams (PI-diagrams) extend Venn diagrams to represent synergy.  
A PI-diagram is composed of nonnegative \emph{partial information regions} (PI-regions).  Unlike the standard Venn entropy diagram in which the sum of all regions is the joint entropy $\ent{\X, Y}$, in PI-diagrams the sum of all regions (i.e. the space of the PI-diagram) is the mutual information $\info{\X}{Y}$.  PI-diagrams show how the mutual information $\info{\X}{Y}$ is distributed across subsets of the predictors.
For example, in the PI-diagram for $n=2$ (\figref{fig:tutorial}): $\{1\}$ denotes the unique information about $Y$ that only $X_1$ carries (likewise \{2\} denotes the information only $X_2$ carries);  $\{1,2\}$ denotes the redundant information about $Y$ that $X_1$ as well as $X_2$ carries, while $\{12\}$ denotes the information about $Y$ that is specified only by  $X_1$ and $X_2$ synergistically or jointly.

\begin{figure}[h!bt]
	\centering
	\subfloat[$\info{X_1}{Y}$]{ \includegraphics[width=0.86in]{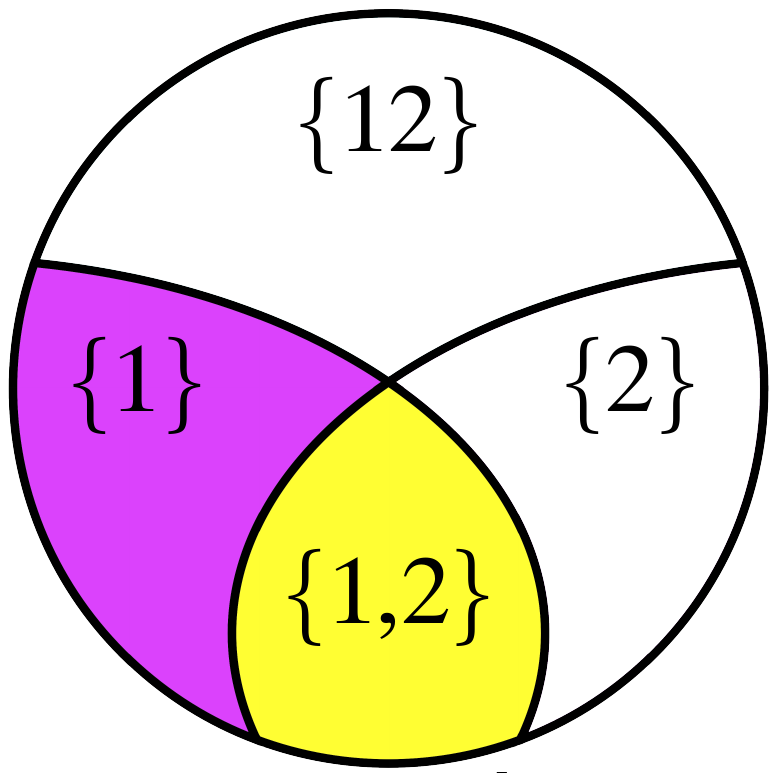} \label{fig:tutorial_a} }
	\subfloat[$\info{X_2}{Y}$]{ \includegraphics[width=0.86in]{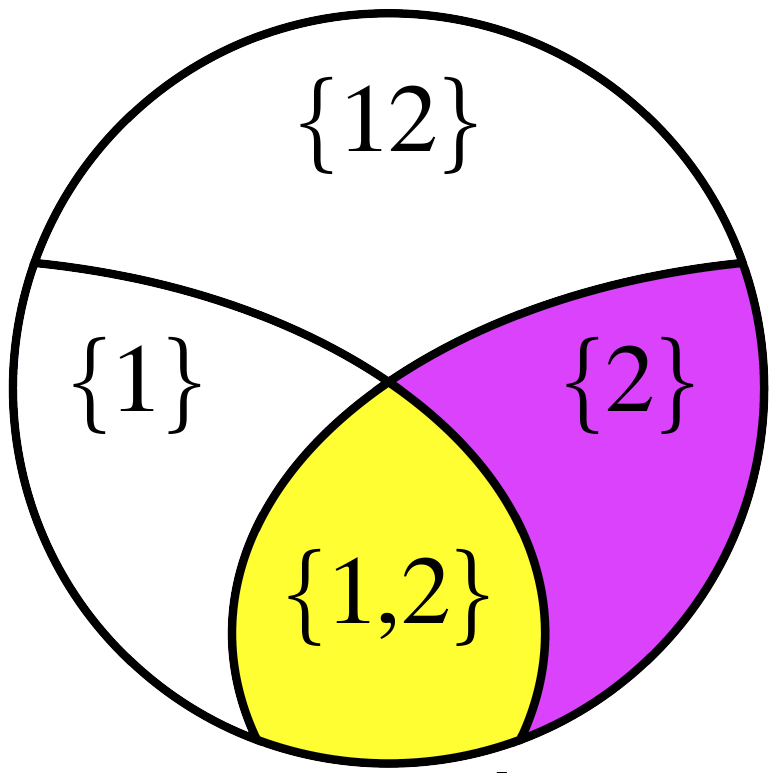} \label{fig:tutorial_b} }
	
	\subfloat[$\info{X_1}{Y|X_2}$]{ \includegraphics[width=0.86in]{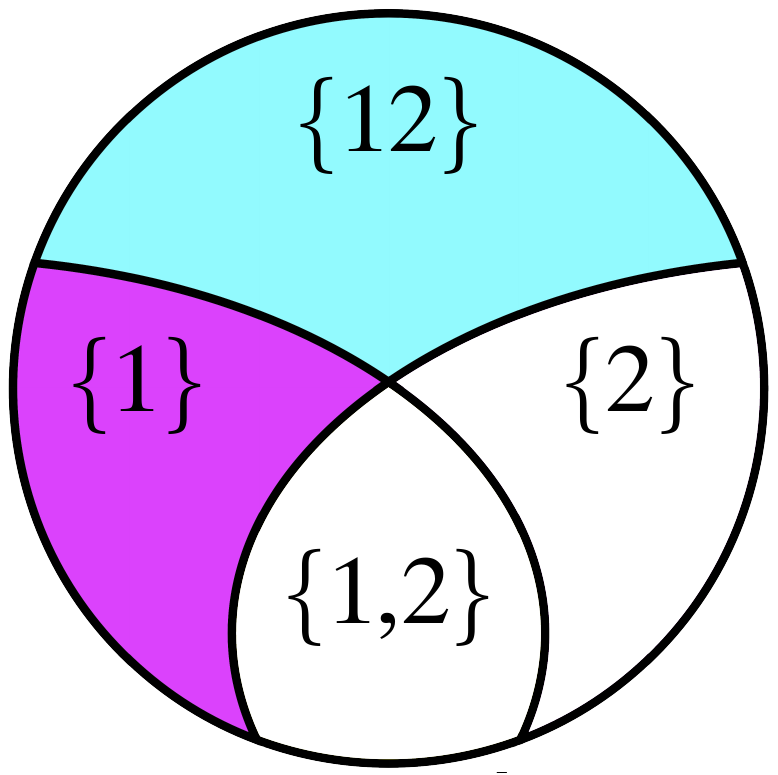} \label{fig:tutorial_c} }
	\subfloat[$\info{X_2}{Y|X_1}$]{ \includegraphics[width=0.86in]{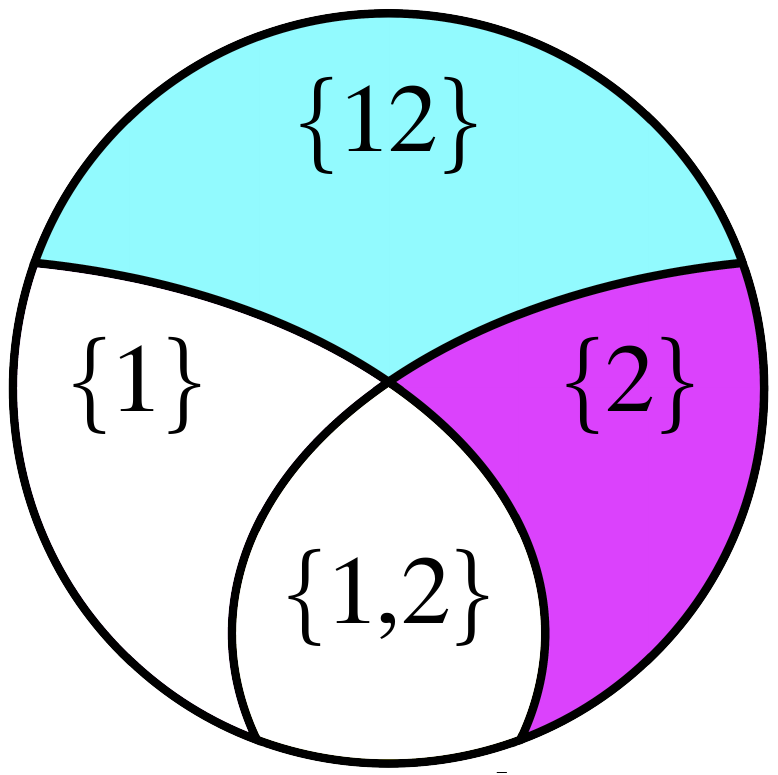} \label{fig:tutorial_d} }
	\subfloat[$\info{X_1 X_2}{Y}$]{ \includegraphics[width=0.86in]{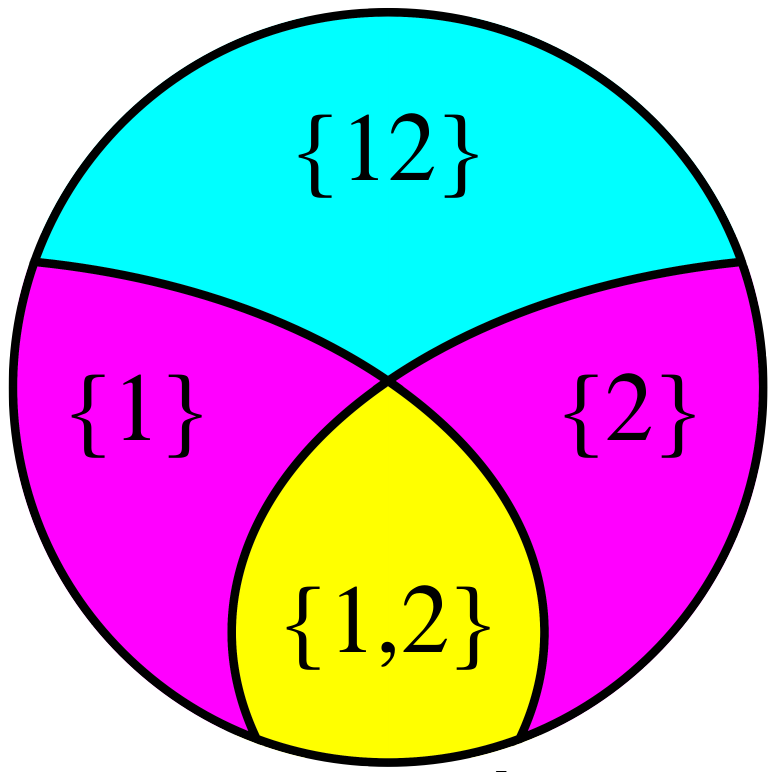} \label{fig:tutorial_e} }
		\caption{PI-diagrams for $n=2$ predictors, showing the amount of redundant (yellow/bottom), unique (magenta/left and right) and synergistic (cyan/top) information with respect to the target $Y$.}
	\label{fig:tutorial}
\end{figure}

Each PI-region is either redundant, unique, or synergistic, but any combination of positive PI-regions may be possible.
Per \cite{plw-10}, for two predictors, the four partial informations are defined as follows: the redundant information as $\Icape{X_1,X_2}{Y}$,
    the unique informations as
\begin{equation}
\begin{split}
    \opI_{\partial}( X_1 : Y ) &= \info{X_1}{Y} - \Icape{X_1,X_2}{Y} \\
    \opI_{\partial}( X_2 : Y ) &= \info{X_2}{Y} - \Icape{X_1,X_2}{Y}, 
\end{split}\label{eq:partialinfos}
\end{equation}and the synergistic information as \begin{equation}\begin{split}
    &\opI_{\partial}( X_1 \vee X_2 : Y )\\
    &= \info{X_1 \vee X_2}{Y} - \opI_{\partial}\left(X_1 :Y \right) - \opI_{\partial}\left(X_2 :Y \right) \\
    &\hspace{0.1 in}- \opI_{\partial}\left(X_1, X_2 :Y \right) \; \\
     &= \info{X_1 \vee X_2}{Y} - \info{X_1}{Y} - \info{X_2}{Y} + \Icape{X_1,X_2}{Y} .
     \end{split}
\label{eq:partialinfos2}
\end{equation}

\section{Desired $\Icap$ properties and canonical examples}
\label{subsec:prop}

There are a number of intuitive properties, proposed in~\cite{plw-10,qsmi,polani12,bertschinger12,lizier13,Iwedge}, that are considered desirable for the intersection information measure $\Icap$ to satisfy:

\begin{itemize}

    \item[\Szero] Weak Symmetry:
$\Icape{X_1,\ldots,X_n}{Y}$ is invariant under reordering of $X_1, \ldots, X_n$.

    \item[\Mzero] Weak Monotonicity:  $\Icape{X_1,\ldots,X_n, Z}{Y} \leq \Icape{X_1,\ldots,X_n}{Y}$ with equality if there exists $X_i \in \{X_1, \ldots, X_n\}$ such that $\ent{Z \vee X_i} = \ent{Z}$.


    \item[\SR] Self-Redundancy: $\Icape{X_1}{Y} = \info{X_1}{Y}$. The intersection information a single predictor $X_1$ conveys about the target $Y$ is equal to the mutual information between the $X_1$ and the target $Y$.

    \item[\Mone] Strong Monotonicity:
$\Icape{X_1,\ldots,X_n, Z}{Y} \leq \Icape{X_1,\ldots,X_n}{Y}$ with equality if there exists $X_i \in \{X_1, \ldots, X_n\}$ such that $\info{Z \vee X_i}{Y} = \info{Z}{Y}$.

    \item[\LP] Local Positivity:  For all $n$, the derived ``partial informations'' defined in \cite{plw-10} are nonnegative.  This is equivalent to requiring that $\Icap$ satisfy \emph{total monotonicity}, a stronger form of supermodularity.  For $n=2$ this can be concretized as, $\Icape{X_1,X_2}{Y} \geq \info{X_1}{X_2} - \info{X_1}{X_2 \middle| Y}$.

    \item[\TM] Target Monotonicity: If $\ent{Y|Z}=0$, then $\Icape{X_1,\ldots,X_n}{Y} \leq \Icape{X_1,\ldots,X_n}{Z}$.


\end{itemize}

There are also a number of canonical examples for which one or more of the partial informations have intuitive values, which are considered desirable for the intersection information measure $\Icap$ to  attain. 

\textbf{Example \textsc{Unq}}, shown in Figure~\ref{fig:exampleU}, is a canonical case of unique information, in which each predictor carries independent information about the target.  $Y$ has four equiprobable states: \bin{ab}, \bin{aB}, \bin{Ab}, and \bin{AB}.  $X_1$ uniquely specifies bit \bin{a}/\bin{A}, and $X_2$ uniquely specifies bit \bin{b}/\bin{B}.  Note that the states are named so as to highlight the two bits of unique information; it is equivalent to choose any four unique names for the four states.

\begin{figure}[h!bt]
	\centering
	\begin{minipage}[c]{0.47\linewidth} \centering
	\subfloat[$\Prob{x_1,x_2,y}$]{\begin{tabular}{ c | c c} \cmidrule(r){1-2}
	$X_1$ $X_2$ &$Y$ \\
	\cmidrule(r){1-2}
	\bin{a b} & \bin{ab} & \quad \nicefrac{1}{4}\\
	\bin{a B} & \bin{aB} & \quad \nicefrac{1}{4}\\
	\bin{A b} & \bin{Ab} & \quad \nicefrac{1}{4}\\
	\bin{A B} & \bin{AB} & \quad \nicefrac{1}{4}\\
	\cmidrule(r){1-2}
	\end{tabular}	
	} \end{minipage}
    \begin{minipage}[c]{0.47\linewidth} \centering
        \begin{align*}
                \info{X_1\vee X_2}{Y} &= 2 \\
                \info{X_1}{Y} &= 1 \\
                \info{X_2}{Y} &= 1 \\
        \end{align*}
        \end{minipage}
        \\[1.5em]
	\begin{minipage}[c]{0.47\linewidth} \centering
	\subfloat[circuit diagram]{ \includegraphics[height=1.3in]{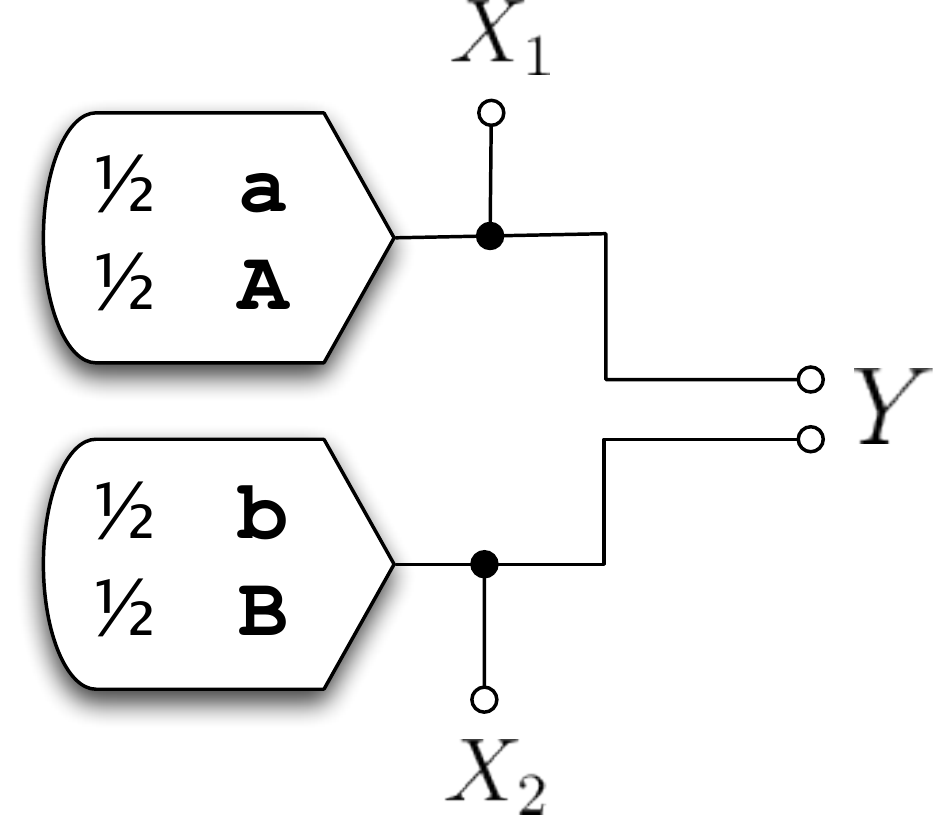} \label{fig:exampleUb} }
	\end{minipage}
	
        \begin{minipage}[c]{0.47\linewidth} \centering
                \subfloat[$\opI_{\min}$]{ \includegraphics[width=1.2in]{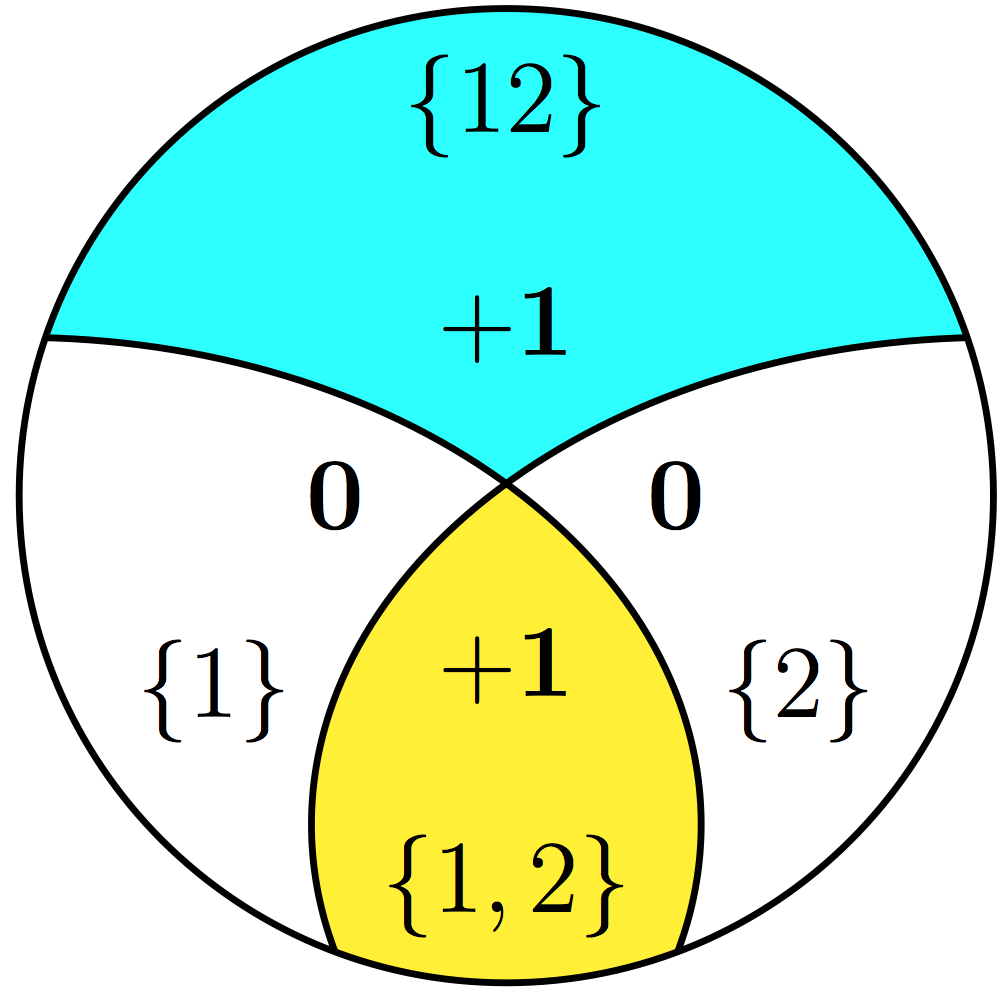} }
        \end{minipage}	
	\begin{minipage}[c]{0.47\linewidth} \centering
		\subfloat[Syn/$\Delta \opI$/$\Iw$/$\Ialpha$]{ \includegraphics[width=1.2in]{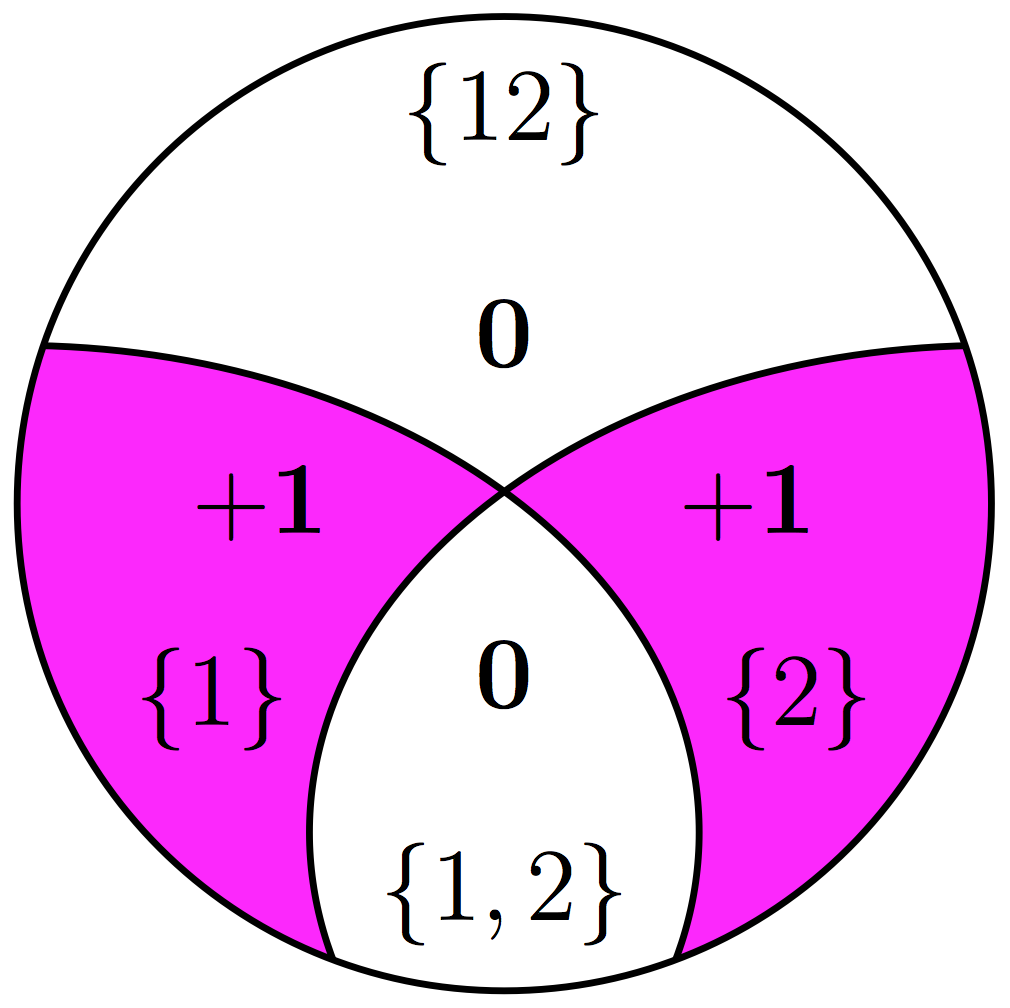} }
	\end{minipage}
	\caption{Example \textsc{Unq}.  $X_1$ and $X_2$ each uniquely carry one bit of information about $Y$.  $\info{X_1 X_2}{Y}~=~\ent{Y}~=~2$ bits.}
	\label{fig:exampleU}
\end{figure}

\textbf{Example RdnXor}, shown in \figref{fig:RdnXor}, is a canonical example of redundancy and synergy coexisting. The  \bin{r}/\bin{R} bit is redundant, while the 0/1 bit of $Y$ is synergistically specified as the XOR of the corresponding bits in $X_1$ and $X_2$.

\begin{figure}
        \centering
        \begin{minipage}[c]{0.55\linewidth} \centering \subfloat[$\Prob{x_1,x_2,y}$]{ \begin{tabular}{ c | c c } \cmidrule(r){1-2}
        $X_1$ $X_2$ & $Y$ \\
        \cmidrule(r){1-2}
        \bin{r0 r0} & \bin{r0} & \quad \nicefrac{1}{8}\\
        \bin{r0 r1} & \bin{r1} & \quad \nicefrac{1}{8}\\
        \bin{r1 r0} & \bin{r1} & \quad \nicefrac{1}{8}\\
        \bin{r1 r1} & \bin{r0} & \quad \nicefrac{1}{8}\\
        \addlinespace
        \bin{R0 R0} & \bin{R0} & \quad \nicefrac{1}{8}\\
        \bin{R0 R1} & \bin{R1} & \quad \nicefrac{1}{8}\\
        \bin{R1 R0} & \bin{R1} & \quad \nicefrac{1}{8}\\
        \bin{R1 R1} & \bin{R0} & \quad \nicefrac{1}{8}\\
        \cmidrule(r){1-2}
        \end{tabular} }
        \end{minipage}
        \begin{minipage}[c]{0.4\linewidth} \centering
        \begin{align*}
                \info{X_1\vee X_2}{Y} &= 2 \\
                \info{X_1}{Y} &= 1 \\
                \info{X_2}{Y} &= 1 \\
                \addlinespace
                \Iminn{X_1,X_2}{Y} &= 1 \\
                \Iwe{X_1,X_2}{Y} &= 1 \\
                \Ialphae{X_1,X_2}{Y} &= 1
        \end{align*}
        \end{minipage}
        \\[1.5em]

        \begin{minipage}[c]{\linewidth} \centering
        \subfloat[circuit diagram]{ \includegraphics[width=2.03in]{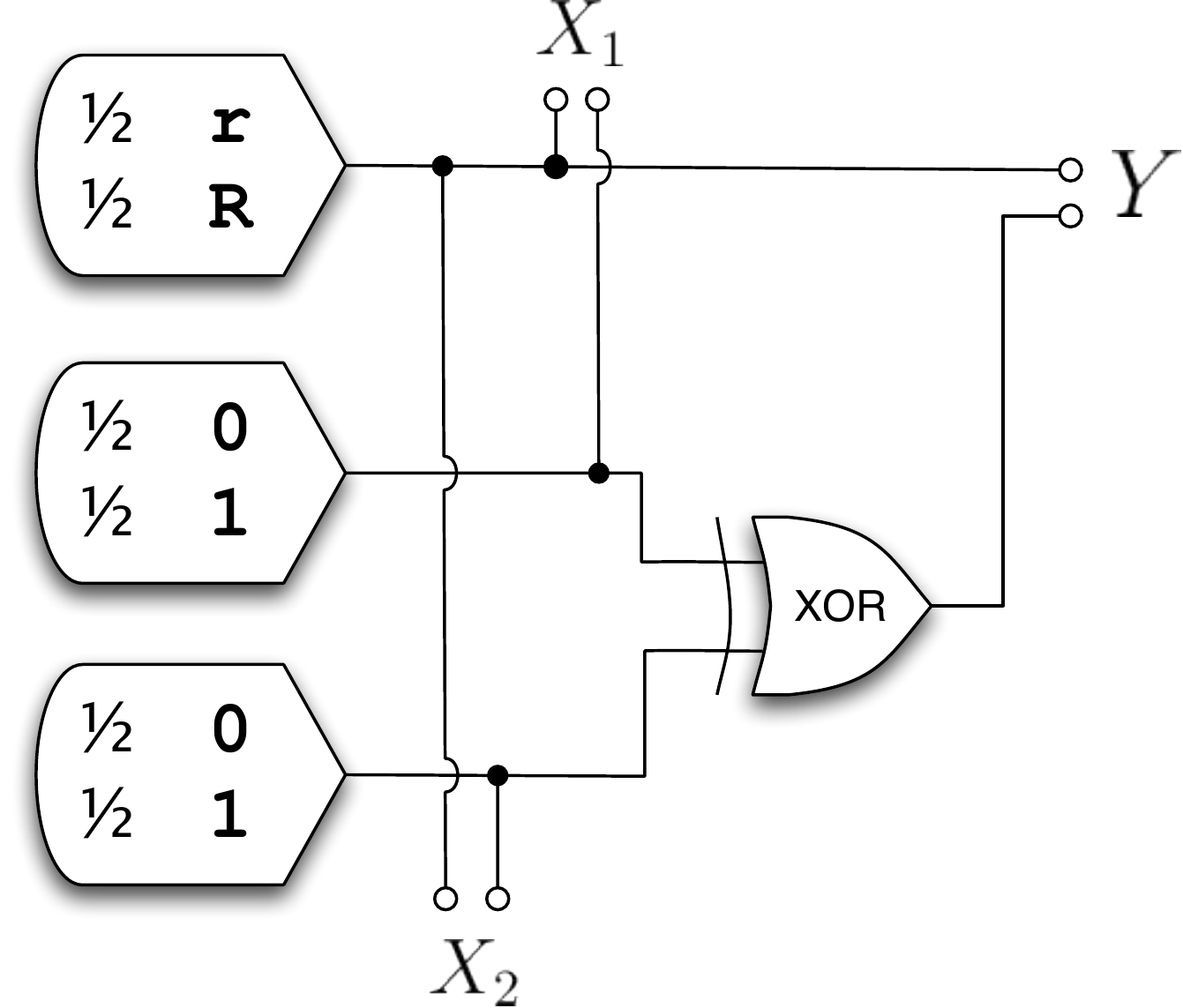} }
        \end{minipage}

        \begin{minipage}[c]{0.47\linewidth} \centering
                \subfloat[Syn]{ \includegraphics[width=1.2in]{PID2-UNQ.png} }
        \end{minipage}
        \begin{minipage}[c]{0.47\linewidth} \centering
                \subfloat[$\Delta \opI$/$\opI_{\min}$/$\Iw$/$\Ialpha$]{ \includegraphics[width=1.2in]{PID2-RDN_XOR.png} }
        \end{minipage}
        \caption{Example \textsc{RdnXor}.  This is the canonical example of redundancy and synergy coexisting.  $\opI_{\min}$ and $\Iw$ each reach the desired decomposition of one bit of redundancy and one bit of synergy.  This example demonstrates $\Iw$ correctly extracting the embedded redundant bit within $X_1$ and $X_2$.}
        \label{fig:RdnXor}
\end{figure}

\textbf{Example And}, shown in \figref{fig:AND}, is an example where the relationship between $X_1,X_2$ and $Y$ is nonlinear, making the desired partial information values less intuitively obvious. Nevertheless, it is desired that the partial information values should be nonnegative. 

\begin{figure}[h!bt]
	\centering
	\begin{minipage}[c]{0.47\linewidth} \centering
	\subfloat[$\Prob{x_1,x_2,y}$]{\begin{tabular}{ c | c c} \cmidrule(r){1-2}
	$X_1$ $X_2$ &$Y$ \\
	\cmidrule(r){1-2}
	\bin{0 0} & \bin{0} & \quad \nicefrac{1}{4}\\
	\bin{0 1} & \bin{0} & \quad \nicefrac{1}{4}\\
	\bin{1 0} & \bin{0} & \quad \nicefrac{1}{4}\\
	\bin{1 1} & \bin{1} & \quad \nicefrac{1}{4}\\
	\cmidrule(r){1-2}
	\end{tabular}	
	} \end{minipage}
    \begin{minipage}[c]{0.47\linewidth} \centering
        \begin{align*}
                \info{X_1\vee X_2}{Y} &= 0.811 \\
                \info{X_1}{Y} &= 0.311 \\
                \info{X_2}{Y} &= 0.311 \\
                \addlinespace
                \Iminn{X_1,X_2}{Y} &= 0.311 \\
                \Iwe{X_1,X_2}{Y} &= 0 \\
                \Ialphae{X_1,X_2}{Y} &= 0
        \end{align*}
        \end{minipage}
        \\[1.5em]
	\begin{minipage}[c]{\linewidth} \centering
	\subfloat[circuit diagram]{ \includegraphics[height=1.4in]{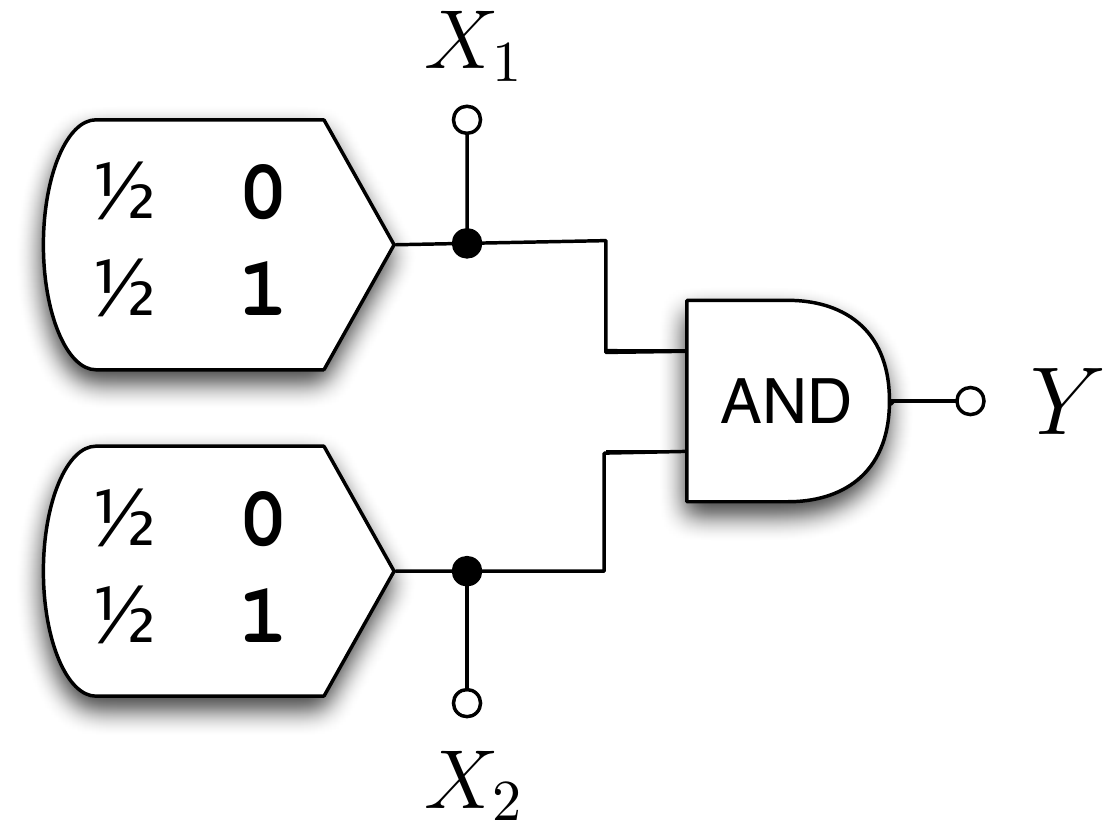} \label{fig:ANDb} }
	\end{minipage}
	
	\begin{minipage}[c]{0.33\linewidth} \centering
                \subfloat[$\Delta \opI$]{ \includegraphics[width=1.2in]{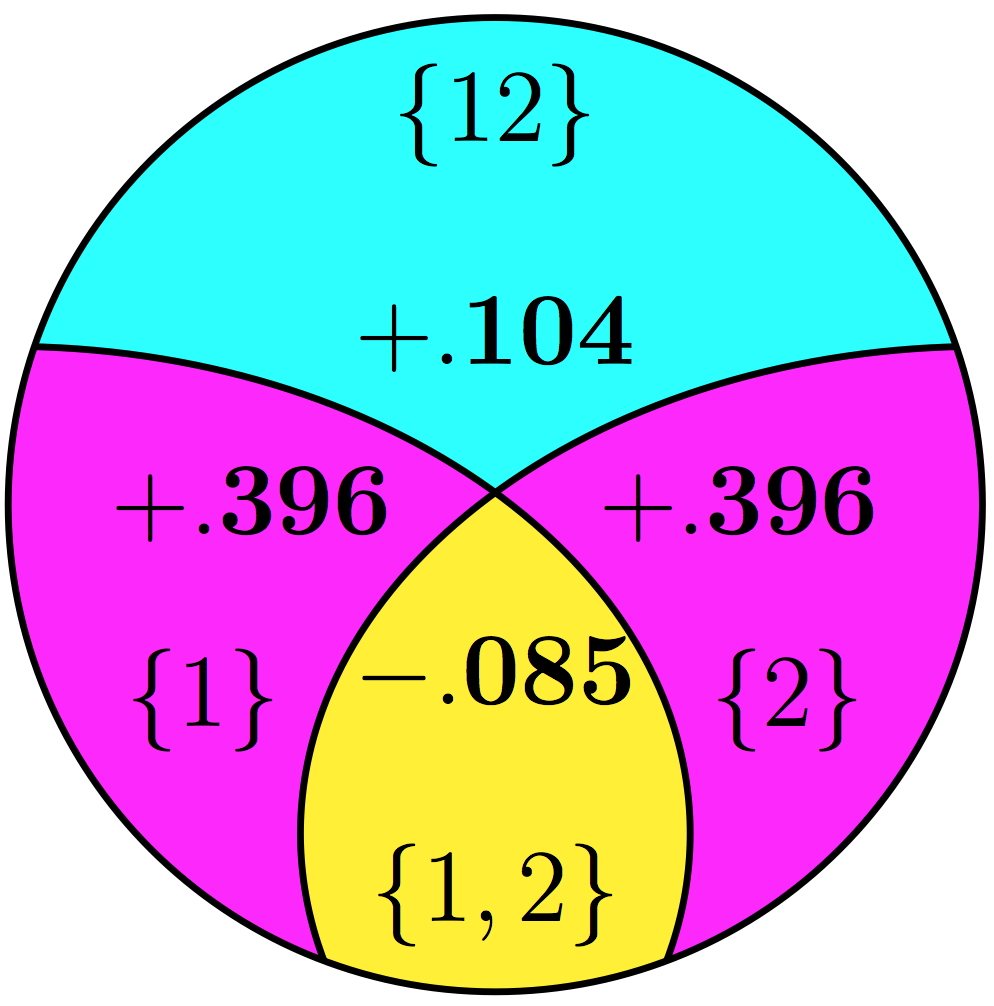} \label{fig:ANDdeltaI} }	
	\end{minipage}
        \begin{minipage}[c]{0.32\linewidth} \centering
		\subfloat[Syn/$\Iw$/$\Ialpha$]{ \includegraphics[width=1.2in]{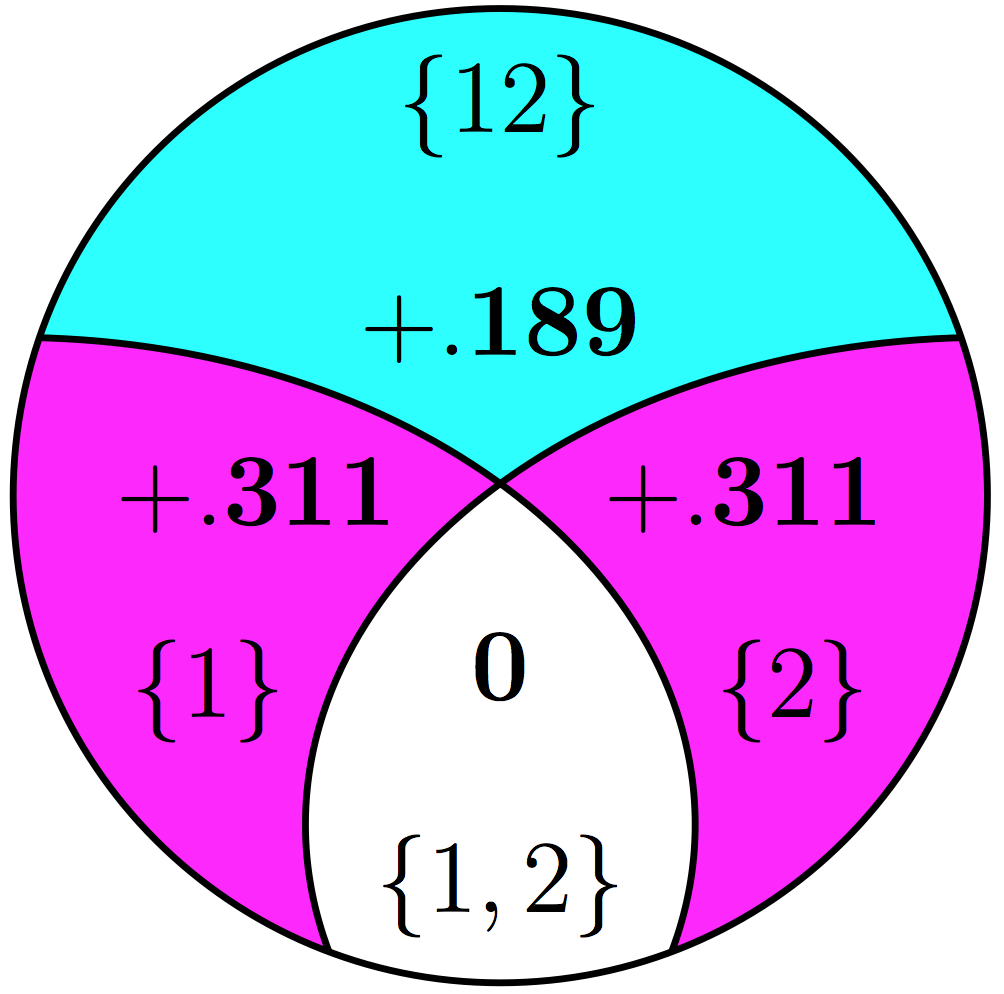} }
        \end{minipage}
    \begin{minipage}[c]{0.32\linewidth} \centering
                \subfloat[$\opI_{\min}$]{ \includegraphics[width=1.2in]{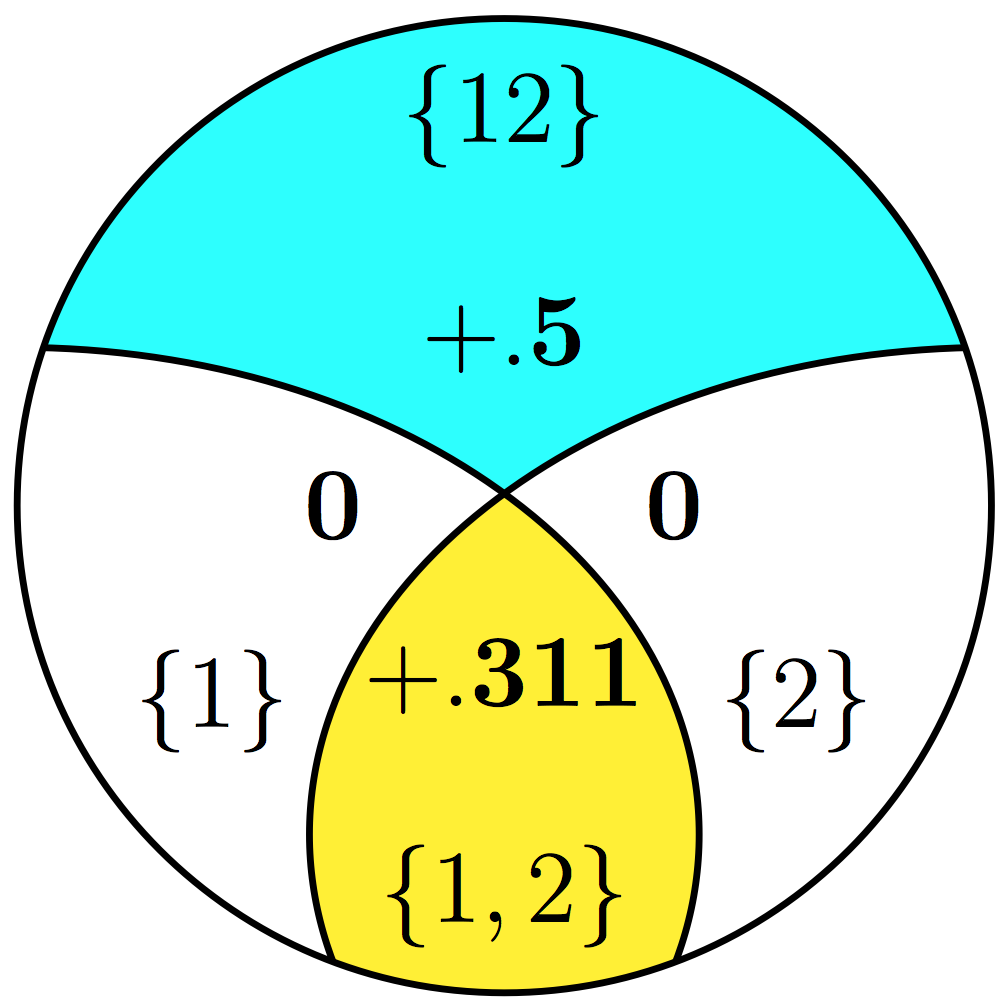} }
        \end{minipage}	
	\caption{Example \textsc{And}.  $\Delta \opI$ yields a redundant information of negative $.085$ bits.  There are arguments for the redundant information being zero or positive, but thus far all that is agreed upon is that the redundant information is between $[0,0.311]$ bits.}
	\label{fig:AND}
\end{figure}

\textbf{Example ImperfectRdn}, shown in \figref{fig:ImperfectRdn}, is an example of ``imperfect'' or ``lossy''  correlation between the predictors, where it is intuitively desirable that the derived redundancy should be positive.
Given \LP, we can determine the desired decomposition analytically.  First, $\info{X_1\vee X_2}{Y} = \info{X_1}{Y} = 1$ bit; therefore, $\info{X_2}{Y|X_1} = \info{X_1 \vee X_2}{Y} - \info{X_1}{Y}~=~0$ bits.  This determines two of the partial informations---the synergistic information $\opI_{\partial}\!\left(X_1 \vee X_2 \:Y\right)$ and the unique information $\opI_{\partial}\!\left( X_2 \:Y\right)$ are both zero.  Then, the redundant information $\opI_{\partial}\!\left( X_1, X_2 \:Y\right) = \info{X_2}{Y} - \opI_{\partial}( X_2 : Y ) = \info{X_2}{Y} = 0.99$ bits.  Having determined three of the partial informations, we compute the final unique information $\opI_{\partial}\!\left( X_1 \:Y\right)=\info{X_1}{Y} - 0.99 = 0.01$ bits.

\begin{figure}
        \centering
        \begin{minipage}[c]{0.47\linewidth} \centering \subfloat[$\Prob{x_1, x_2, y}$]{ \begin{tabular}{ c | c c } \cmidrule(r){1-2}
$\ \, X_1 \, X_2$  &$Y$ \\
\cmidrule(r){1-2}
\bin{0 0} & \bin{0} & \quad $0.499$\\
\bin{0 1} & \bin{0} & \quad $0.001$\\
\bin{1 1} & \bin{1} & \quad $0.500$\\
\cmidrule(r){1-2}
\end{tabular}
\label{fig:ANDa} }
\end{minipage} \begin{minipage}[c]{0.47\linewidth} \centering
        \begin{align*}
        \info{X_1\vee X_2}{Y} &= 1 \\
        \info{X_1}{Y} &= 1 \\
        \info{X_2}{Y} &= 0.99 \\
        \end{align*}
        \end{minipage}
        \\[1.5em]
        \begin{minipage}[c]{0.9\linewidth} \centering
        \subfloat[circuit diagram]{ \includegraphics[width=1.95in]{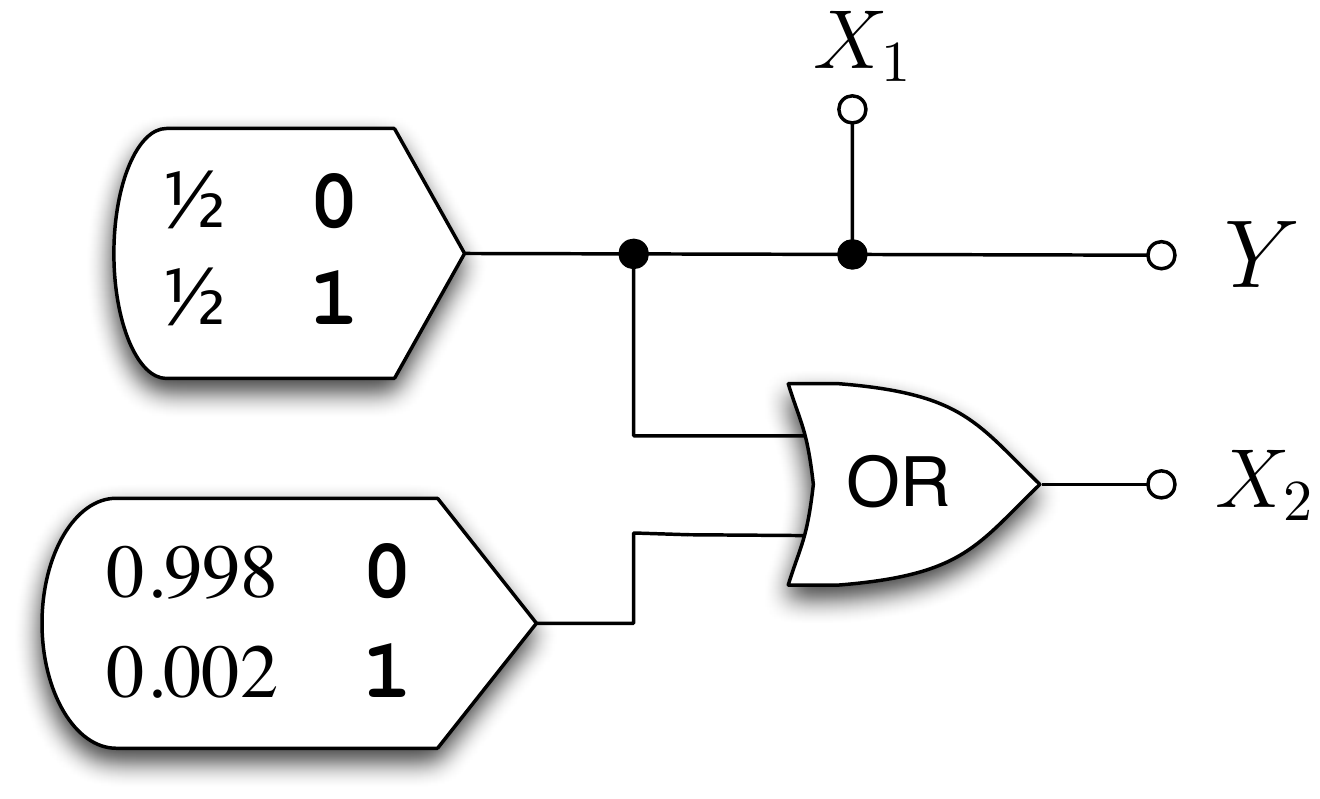} }
    \end{minipage}

        \begin{minipage}[c]{0.47\linewidth}
        \subfloat[$\Iw$]{ \includegraphics[width=1.2in]{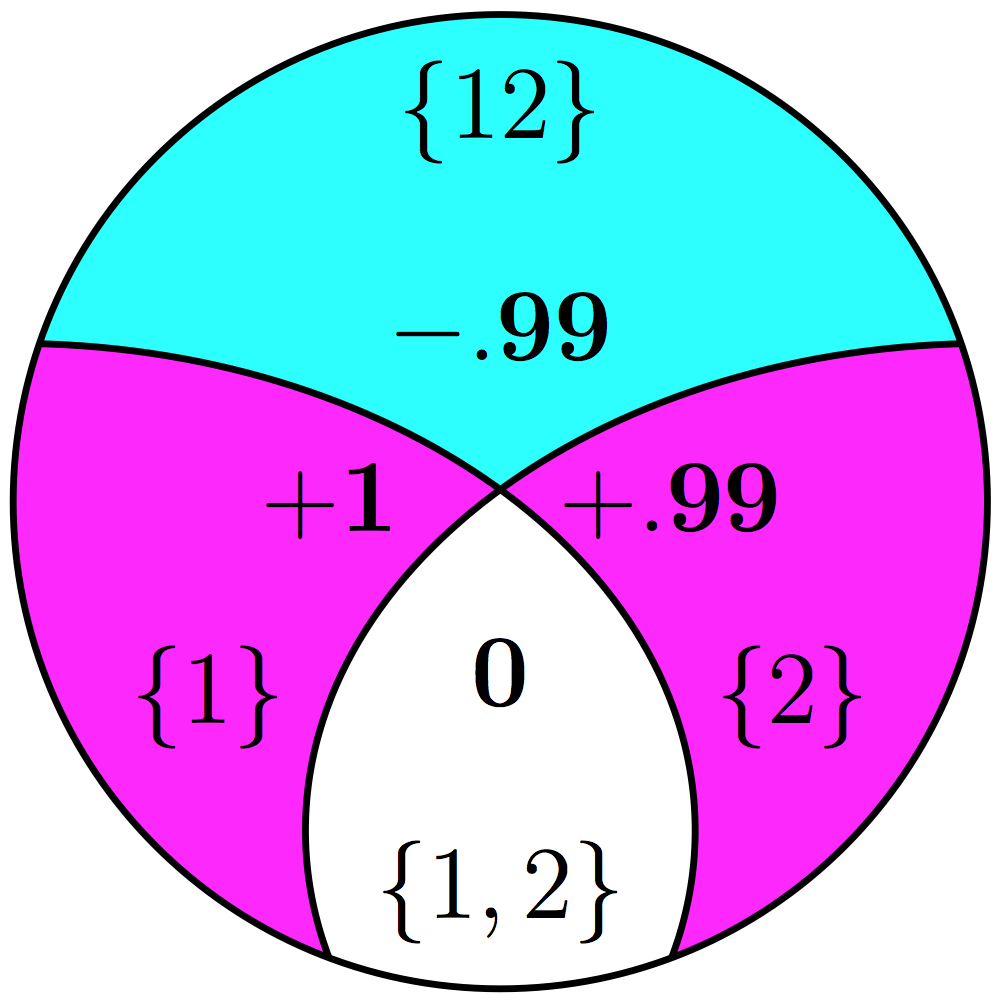} \label{fig:ImperfectRdn_Iw} }
    \end{minipage}
        \begin{minipage}[c]{0.47\linewidth}
        \subfloat[Syn/$\Delta \opI$/$\opI_{\min}$/$\Ialpha$]{ \includegraphics[width=1.2in]{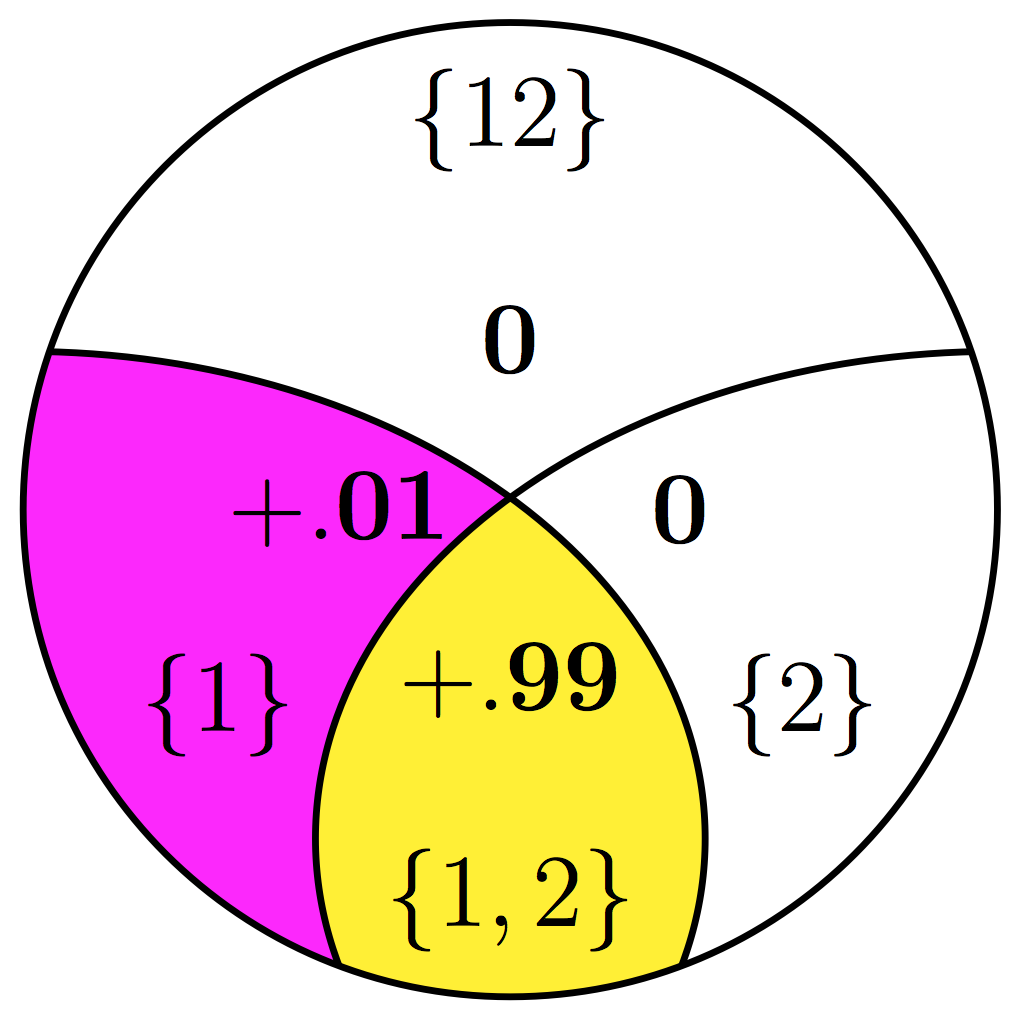} \label{fig:ImperfectRdn_Imin} }
    \end{minipage}
        \caption{Example \textsc{ImperfectRdn}. $\Iw$ is blind to the noisy correlation between $X_1$ and $X_2$ and calculates zero redundant information.  An ideal $\Icap$ measure would detect that all of the information $X_2$ specifies about $Y$ is also specified by $X_1$ to calculate $\Icape{X_1, X_2}{Y} = 0.99$ bits.}
        \label{fig:ImperfectRdn}
\end{figure}

\section{Previous candidate measures}

In \cite{plw-10}, the authors propose to use the following quantity,
$\opI_{\min}$, as the intersection information measure:
\begin{equation}
\begin{split}
    &\opI_{\min}\left( X_1, \ldots, X_n : Y \right)\\ &\equiv \sum_{y \in Y} \Prob{y} \min_{i \in \{1, \ldots, n\}} \info{X_i}{Y=y} \\
     &= \sum_{y \in Y} \Prob{y} \min_{i \in \{1, \ldots, n\}} \DKL{ \Prob{X_i|y} }{ \Prob{X_i} } \; ,
\end{split}
\label{eq:Imin_DKL}
\end{equation}
where $\opname{D_{KL}}$ is the Kullback-Leibler divergence.

Though $\opI_{\min}$ is an intuitive and plausible choice for the intersection information, \cite{qsmi} showed that $\opI_{\min}$ has counterintuitive properties. In particular, $\opI_{\min}$ calculates one bit of redundant information for example \textsc{Unq} (\figref{fig:exampleU}). It does this because each input shares one bit of information with the output.  However, its quite clear that the shared informations are, in fact, different: $X_1$ provides the low bit, while $X_2$ provides the high bit. This led to the conclusion that $\opI_{\min}$ \emph{overestimates} the ideal intersection information measure by focusing only on \emph{how much} information the inputs provide to the output.
Another way  to understand why $\opI_{\min}$ overestimates redundancy in example \textsc{Unq} is to imagine a hypothetical example where there are exactly two bits of unique information for every state $y \in Y$ and no synergy or redundancy. $\opI_{\min}$ would calculate the redundancy as the minimum over both predictors which would be $\min[1, 1] = 1$ bit. Therefore $\opI_{\min}$ would calculate $1$ bit of redundancy even though by definition there was no redundancy but merely two bits of unique information.

A candidate measure of synergy, $\Delta \opI$, proposed in \cite{nirenberg03}, leads to a negative value of redundant information for Example \textsc{And}. Starting from $\Delta \opI$ as a direct measure for synergistic information and then using eqs.~\eqref{eq:partialinfos} and \eqref{eq:partialinfos2} to derive the other terms, we get \figref{fig:ANDdeltaI} showing $\Icape{X_1,X_2}{Y} \approx -.085$ bits.

Another candidate measure of synergy,  Syn \cite{schneidman-03}, calculates zero synergy and redundancy for Example \textsc{RdnXor}, as opposed to the intuitive value of one bit of redundancy and one bit of synergy.

\section{New candidate measures}
\subsection{The $\Iw$ measure}
\label{sect:Icrv}
Based on \cite{wolf04}, we can consider a candidate intersection information as the maximum mutual information $\info{Q}{Y}$ that some random variable $Q$ conveys about $Y$, subject to $Q$ being a function of each predictor $X_1,\ldots,X_n$.  After some algebra, the leads to,

\begin{align}
\label{eq:crvcap}
\begin{split}
    &\Iwe{X_1, \ldots, X_n}{Y} \\&\equiv \max_{ \Pr(Q|Y) } \info{Q}{Y} \\
    &\hspace{0.2in} \textnormal{subject to } \forall i \in \{1, \ldots, n\} : \ent{Q|X_i} = 0
\end{split} \; ,
\end{align}
which reduces to a simple expression in \cite{Iwedge}.

Example \textsc{ImperfectRdn} highlights the foremost shortcoming of $\Iw$; $\Iw$ does not detect ``imperfect'' or ``lossy'' correlations between $X_1$ and $X_2$.  Instead, $\Iw$ calculates  zero redundant information, that $\Icape{X_1,X_2}{Y}=0$ bits.  This arises from $\Prob{X_1 = \bin{1}, X_2 = \bin{0}} > 0$.  If this were zero, \textsc{ImperfectRdn} reverts to being determined by the properties \SR and the \Mzero equality condition.  Due to the nature of the common random variable, $\Iw$ only sees the ``deterministic'' correlations between $X_1$ and $X_2$---add even an iota of noise between $X_1$ and $X_2$ and $\Iw$ plummets to zero. This highlights a related issue with $\Iw$; it is not continuous---an arbitrarily small change in the probability distribution can result in a discontinuous jump in the value of $\Iw$.

Despite this, $\Iw$ is useful stepping-stone, it captures what is inarguably redundant information (the common random variable).  In addition, unlike earlier measures, $\Iw$ satisfies \TM.

\subsection{The $\Ialpha$ measure}
\label{sect:Ialpha}
Intuitively, we expect that if $Q$ only specifies redundant information, that conditioning on any predictor $X_i$ would vanquish all of the information $Q$ conveys about $Y$.  Noticing that $\Iw$ \emph{underestimates} the ideal $\Icap$ measure (i.e. it doesn't satisfy \LP), we loosen the constraint $\ent{Q|X_i}=0$ in eq.~\eqref{eq:crvcap}, leading us to define the measure $\Ialpha$:
\begin{align}
\begin{split}
    &\Ialphae{X_1, \ldots, X_n}{Y} \\
    &\equiv \max_{ \Pr(Q|Y) } \info{Q}{Y} \\
    &\hspace{0.2in} \textnormal{subject to } \forall i \in \{1, \ldots, n\} : \info{Q \vee X_i}{Y}=\info{X_i}{Y} \; .
\end{split} \\
\begin{split}
     &= \max_{ \Pr(Q|Y) } \info{Q}{Y} \\
    &\hspace{0.2in} \textnormal{subject to } \forall i \in \{1, \ldots, n\} : \ent{Q\middle|X_i} = \ent{Q\middle|X_i \vee Y} \; .
\end{split}
\end{align}
This measure obtains the desired values for the canonical examples in Section~\ref{subsec:prop}. However, its implicit definition makes it more difficult to verify whether or not it satisfies the desired properties  in Section~\ref{subsec:prop}.  Pleasingly, $\Ialpha$ also satisfies \TM.

We can also show (See Lemmas~\ref{lem:IwleqIalpha} and \ref{lem:IalphaleqImin} in Appendix~\ref{appendix:miscproofs}) that
\begin{align}
\begin{split}
    &0 \leq \Iwe{X_1,\ldots,X_n}{Y} \leq \Ialphae{X_1, \ldots, X_n}{Y} \\&\leq \opI_{\min}\left( X_1,\ldots,X_n : Y \right) \; .
\end{split}
\end{align}

While  $\Ialpha$ satisfies previously defined canonical examples, we have found another example, shown in \figref{fig:Subtle}, for which  $\Iw$ and $\Ialpha$ both calculate negative synergy. This example further complicates Example \textsc{And} by making the predictors mutually dependent.

\begin{figure}
        \centering
        \begin{minipage}[c]{0.45\linewidth} \centering \subfloat[$\Prob{x_1, x_2, y}$]{ \begin{tabular}{ c | c c } \cmidrule(r){1-2}
$\ \, X_1 \, X_2$  &$Y$ \\
\cmidrule(r){1-2}
\bin{0 0} & \bin{00} & \quad \nicefrac{1}{3}\\
\bin{0 1} & \bin{01} & \quad \nicefrac{1}{3}\\
\bin{1 1} & \bin{11} & \quad \nicefrac{1}{3}\\
\cmidrule(r){1-2}
\end{tabular}
\label{fig:NEGSYN} }
\end{minipage} \begin{minipage}[c]{0.45\linewidth} \centering
        \begin{eqnarray*}
        \nonumber \info{X_1\vee X_2}{Y} &= 1.585 \\
        \info{X_1}{Y} &= 0.918 \\
        \info{X_2}{Y} &= 0.918 \\
        \info{X_1}{X_2} &= 0.252
        \end{eqnarray*}
        \end{minipage}
        \\[1.5em]
        \begin{minipage}[c]{\linewidth} \centering
        \subfloat[circuit diagram]{ \includegraphics[width=2.1in]{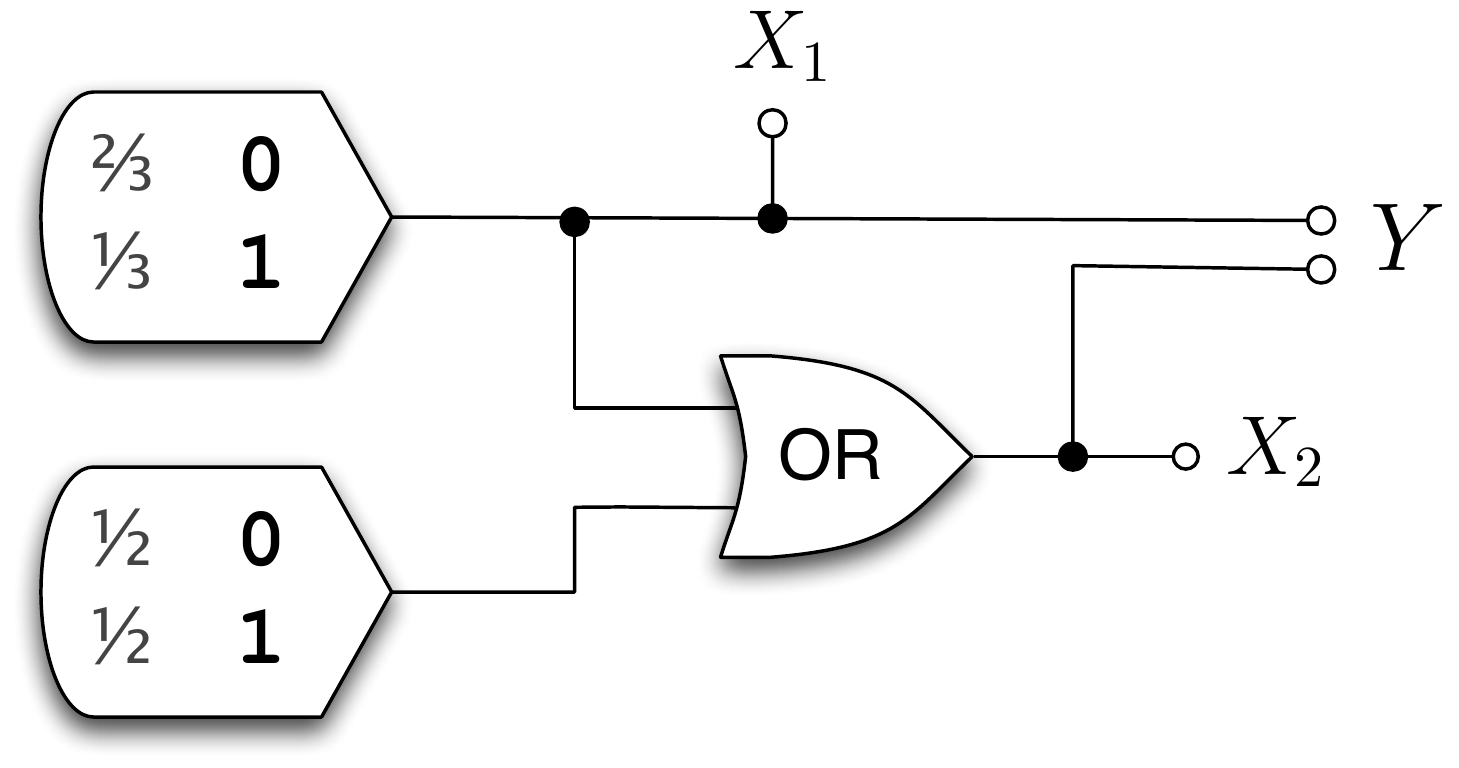} \label{fig:Subtle_circuit} }
    \end{minipage}

        \begin{minipage}[c]{0.45\linewidth} \centering
        \subfloat[$\Iw$/$\Ialpha$]{ \includegraphics[width=1.3in]{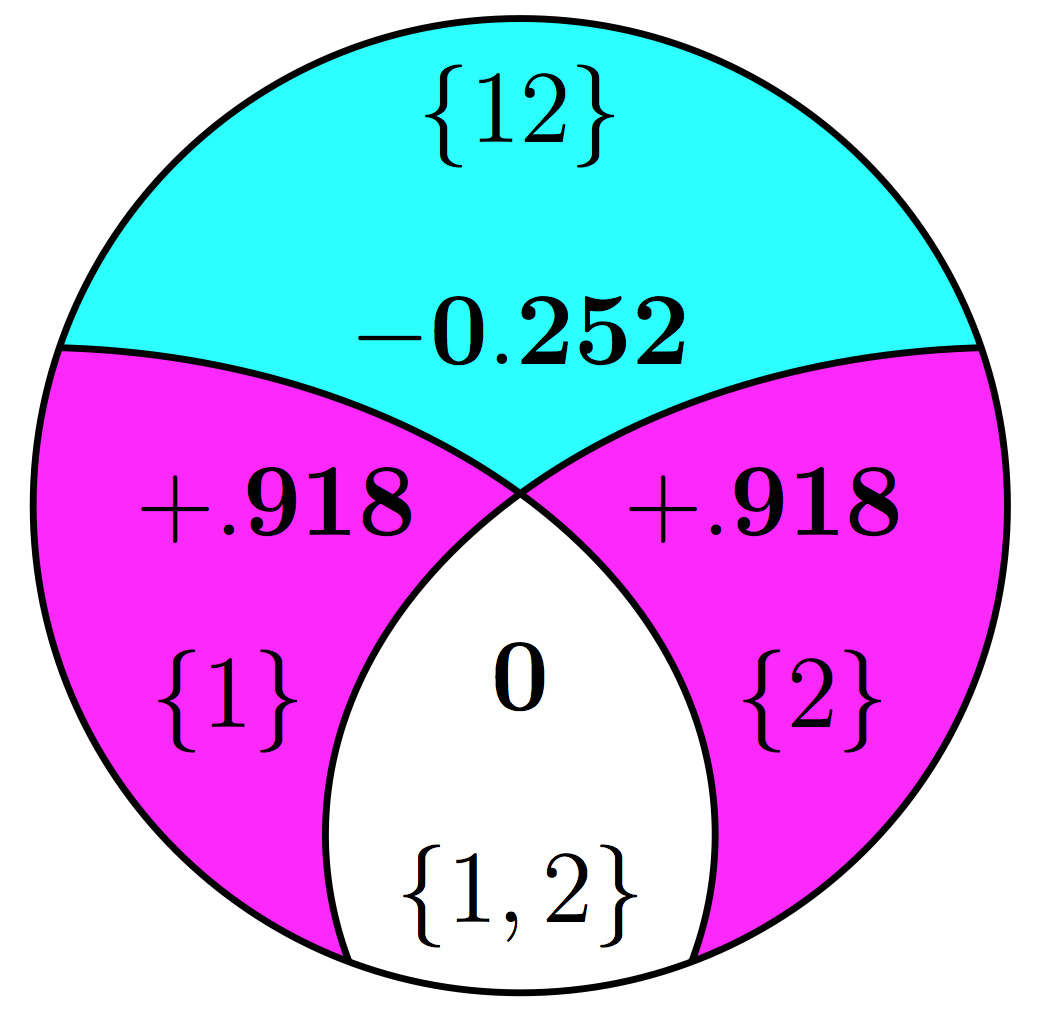} }
    \end{minipage}
        \begin{minipage}[c]{0.45\linewidth} \centering
        \subfloat[$\opI_{\min}$]{ \includegraphics[width=1.25in]{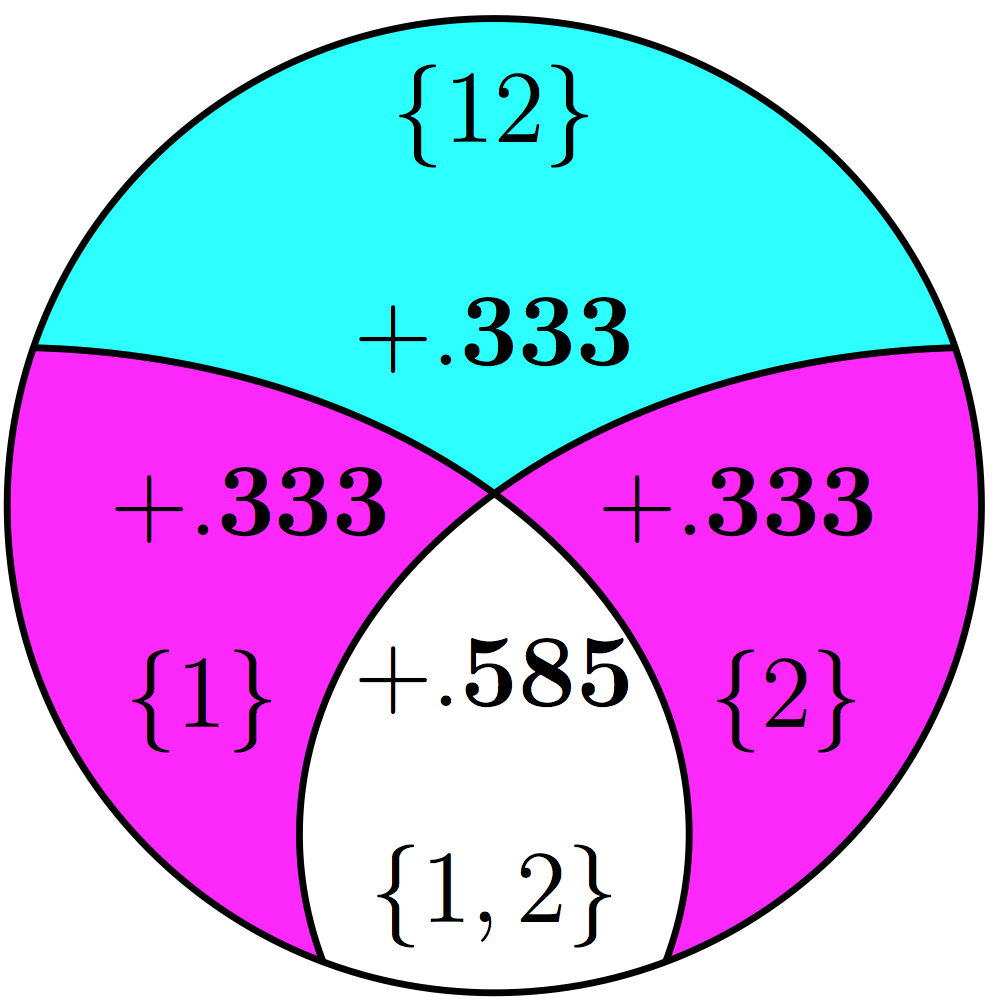} }
    \end{minipage}
        \caption{Example \textsc{Subtle}.}
    \label{fig:Subtle}
\end{figure}

\section{Conclusion}
We have defined new measures for redundant information of predictor random variables regarding a target random variable.
It is not clear whether it is possible for a single measure of synergy/redundancy to satisfy all previously proposed desired properties and canonical examples, and some of them are debatable. For example, a plausible measure of the ``unique information'' \cite{bertschinger13} and ``union information'' \cite{qsmi} yields an equivalent $\Icap$ measure that does not satisfy \TM. Determining whether some of these properties are contradictory is an interesting question for further work.

\textbf{Acknowledgements}.  We thank Jim Beck, Yaser Abu-Mostafa, Edwin Chong, Chris Ellison, and Ryan James for helpful discussions.

\bibliography{quant_synergy}

\appendix
\textbf{Proof $\Ialpha$ does not satisfy \LP}.  Proof by counter-example \textsc{Subtle} (\figref{fig:Subtle}).
For $\info{Q}{Y|X_1}=0$, then $Q$ must not distinguish between states of $Y=\bin{00}$ and $Y=\bin{01}$ (because $X_1$ does not distinguish between these two states).  This entails that $\Prob{Q|Y=\bin{00}}=\Prob{Q|Y=\bin{01}}$.
By symmetry, likewise for $\info{Q}{Y|X_2}=0$, $Q$ must be distinguish between states $Y=\bin{01}$ and $Y=\bin{11}$.  Altogether, this entails that $\Prob{Q|Y=\bin{00}} = \Prob{Q|Y=\bin{01}} = \Prob{Q|Y=\bin{11}}$, which then entails, $\Prob{q|y_i} = \Prob{q|y_j} \quad \forall q \in Q, y_i \in Y, y_j \in Y$, which is only achievable when $\Prob{q} = \Prob{q|y} \quad \forall q \in Q, y \in Y$.  This makes $\info{Q}{Y}=0$, therefore for example \textsc{Subtle}, $\Ialphae{X_1, X_2}{Y}=0$.
%
%
%
\label{appendix:miscproofs}

\begin{lem}\label{lem:IwleqIalpha}
We have $\Iwe{X_1, \ldots, X_n}{Y} \leq  \Ialphae{X_1, \ldots, X_n}{Y}$.
\end{lem}  \begin{proof} We define a random variable $Q^\prime = X_1 \wedge \cdots \wedge X_n$.  We then plugin $Q^\prime$ for $Q$ in the definition of $\Ialpha$.  This newly plugged-in $Q$ satisfies the constraint $\forall i \in \{1, \ldots, n\}$ that $\info{Q}{Y|X_i}=0$.  Therefore, $Q^\prime$ is always a possible choice for $Q$, and the maximization of $\info{Q}{Y}$ in $\Ialpha$ must be at least as large as $\info{Q^\prime}{Y} = \Iwe{X_1, \ldots, X_n}{Y}$.
\end{proof}

\begin{lem} \label{lem:IalphaleqImin}
We have $\Ialphae{X_1, \ldots, X_n}{Y} \leq  \opI_{\min}\left( X_1, \ldots, X_n : Y \right)$
\end{lem}
\begin{proof}
For a given state $y \in Y$ and two arbitrary random variables $Q$ and $X$, given $\info{Q}{y|X}=\DKL{\Prob{QX|y}}{\Prob{Q|X}\Prob{X|y}}=0$, we show that, $\info{Q}{y} \leq \info{X}{y}$,

\begin{eqnarray*}
    \info{X}{y} - \info{Q}{y} &= \sum_{x \in X} \Prob{x|y} \log \frac{\Prob{x|y}}{\Prob{x}} - \sum_{q \in Q} \Prob{q|y} \log \frac{ \Prob{q|y} }{ \Prob{q} } \\
    &\geq& 0 \; .
\end{eqnarray*}

Generalizing to $n$ predictors $X_1, \ldots, X_n$, the above shows that that the maximum $\info{Q}{y}$ under constraint $\info{Q}{y|X_i}$ will always be less than $\min_{i \in \{1, \ldots, n\}} \info{X_i}{y}$, which completes the proof.
\end{proof}

\begin{lem} \label{lem:IminMone}
Measure $\opI_{\min}$ satisfies desired property Strong Monotonicity, \Mone.
\end{lem}
\begin{proof}
Given $\ent{Y|Z}=0$, then the specific-surprise $\info{Z}{y}$ yields,

\begin{eqnarray*}
\info{Z}{y} &\equiv& \DKL{ \Prob{Z|y} }{ \Prob{Z} } \\
    &=& \sum_{z \in Z} \Prob{z|y} \log \frac{ \Prob{z|y} }{ \Prob{z} } \\
    &=& \sum_{z \in Z} \Prob{z|y} \log \frac{1}{\Prob{y}} \\
    &=& \log \frac{1}{ \Prob{y} } \; .
\end{eqnarray*}

Given that for an arbitrary random variable $X_i$, $\info{X_i}{y} \leq \log \frac{1}{\Prob{y}}$.  As $\opI_{\min}$ takes only uses the $\min_i \info{X_i}{y}$, the minimum is invariant under adding any predictor $Z$ such that $\ent{Y|Z}=0$.  Therefore, measure $\opI_{\min}$ satisfies property \Mone.
\end{proof}


\end{document}